\documentclass[12pt]{article}
\usepackage{amsmath,bbm,amssymb}
\usepackage{amsthm}
\usepackage{comment}
\usepackage{graphicx, graphics}
\usepackage{textcomp}
\usepackage{amsfonts}

\usepackage{psfrag,pst-node,subfigure,rotating, amsmath, bbm, amsthm, amssymb, amsthm, setspace, picture, epsfig, amsfonts, upgreek}

\usepackage{multirow}

\usepackage{arydshln}

\newtheorem{assumption}{Assumption}

\newtheorem{theorem}{Theorem}

\newtheorem{proposition}{Proposition}
\newtheorem{example}{Example}

\def\v{\mathbf}
 \def\X{{\mathbb{X}}} \def\P{{\mathrm{Pr}}}

\begin{document}
\begin{center}
    {\Large\bf Local optimization-based statistical inference}
\\[2mm] {\large Shifeng Xiong}
\\ Academy of Mathematics and Systems Science \\Chinese Academy of Sciences, Beijing 100190\\xiong@amss.ac.cn
\end{center}
\vspace{-10mm}

\vspace{1cm} \noindent{\bf Abstract}\quad This paper introduces a local optimization-based approach to test statistical hypotheses and to construct confidence intervals.
This approach can be viewed as an extension of bootstrap, and yields asymptotically valid tests and confidence intervals as long as there exist consistent estimators of
unknown parameters. We present simple algorithms including a neighborhood bootstrap method to implement the approach. Several examples in which theoretical analysis is
not easy are presented to show the effectiveness of the proposed approach.

\vspace{1cm} \noindent{\bf AMS 2000 subject classifications:} 62F03, 62F25, 62F40.

\vspace{2mm} \noindent{{\bf KEY WORDS:} Bootstrap; Frequentist; Importance sampling; Non-regular problem; Resampling; Space-filling design; Stochastic programming.}

\newpage

\section{Introduction}\label{sec:intro}
\hskip\parindent \vspace{-0.8cm}

More and more complex datasets call for sophisticated statistical methods in the modern era. Compared with other fields for analyzing data such as computer science and
applied mathematics, statistics can quantify the uncertainty of a phenomenon via hypothesis testing and/or interval estimation, which solidifies the unique feature of
this discipline. In conventional frenquentist statistics, for testing a hypothesis or constructing a confidence interval, we need to find proper test statistic or
pivotal quantity whose distribution satisfies certain properties (Lehmann and Romano 2006). However, this is quite difficult for many complex problems. The bootstrap
method (Efron 1979) relaxes the above requirement on test statistics or pivotal quantities via its ability in distribution approximation, and thus strengthens the power
of conventional frequentist inference. Another advantage of bootstrap is that it provides explicit resampling-based solutions if the underlying model is well estimated.
Consequently, bootstrap has been well received in statistics and other fields. The frequentist properties of bootstrap inferential procedures such as the bootstrap
interval estimation can be guaranteed by the consistency of bootstrap distribution estimation (Shao and Tu 1995). This is also true for related methods like subsampling
(Politis, Romano, and Wolf 1999). Generally speaking, it is more difficult to prove such a consistency than to derive the asymptotic distribution of the corresponding
test statistic or pivotal quantity.

From the above discussion it can be seen that we have to do much theoretical work before claiming that the proposed method is a frequentist one. This is not easy for
complex problems, and thus hampers the frequentist approach from being more applicable. In this paper we provide a very general approach based on local optimization to
complement current frequentist inference. Our approach can be viewed as an extension of the classical bootstrap method, and reduces to it when the region for
optimization shrinks to the centre. On the theoretical aspect, the tests and confidence intervals constructed by our approach possess asymptotic frequentist properties
as long as we have consistent estimators of unknown parameters. This feature indicates that we do not need to derive any (asymptotic) distribution or to prove the
consistency of distribution estimation before using the proposed approach. In addition, with a proper region for optimization, the proposed approach is first order
asymptotically equivalent to the bootstrap method for regular problems. On the computational aspect, our approach only requires the optimal objective value of an
optimization problem over a local region, which can be reached by standard optimization techniques. We also present simple experimental design-based algorithms including
a neighborhood bootstrap method to solve the optimization problem. These algorithms are easy to implement for practitioners, and produce satisfactory results in our
simulations.

The rest of this paper is organized as follows. Sections \ref{sec:ht} and \ref{sec:ci} introduce local optimization-based hypothesis testing and interval estimation,
respectively. Their asymptotic frequentist properties are studied in Section \ref{sec:ap}. Some implementation issues are discussed in Section \ref{sec:imp}. Section
\ref{sec:exm} presents four non-regular examples including a high-dimensional problem and a nonparametric regression problem to illustrate the proposed approach. We end
the paper with some discussion in Section \ref{sec:diss}.

\section{Local optimization-based hypothesis testing}\label{sec:ht}\hskip\parindent \vspace{-0.8cm}


Let the random sample $\X$ be drawn from a distribution $F(\cdot,\theta)$, where $\theta$ lies in the parameter space $\Theta$. Here $\Theta$ can be a subset of an
Euclidean space or an infinite-dimensional space. We are interested in testing\begin{equation}H_0:\ \theta\in\Theta_0\ \leftrightarrow\ H_1:\
\theta\in\Theta\setminus\Theta_0,\label{H}\end{equation}where $\Theta_0$ is a close subset of $\Theta$. Let $T=T(\X)\in\mathbb{R}$ be a test statistic. Suppose that $T$
tends to take a large value when $H_0$ does not hold. It is known that the $p$-value for testing \eqref{H} is defined as
\begin{equation}P=\sup_{\phi\in\Theta_0}\P(T_\phi^*\geqslant T\mid T),\label{pv}\end{equation}where $T_\phi^*=T(\X^*)$ and $\X^*$ is an independent copy
of $\X$ from $F(\cdot,\phi)$ (Fisher 1959). Given a significance level $\alpha\in(0,1)$, we will reject $H_0$ if $P<\alpha$. This test can strictly control Type I error
within the Neyman-Pearson framework, as shown in the following proposition.

\begin{proposition}\label{prop: fp} Under $H_0$, $$\P(P<\alpha)\leqslant\alpha$$\end{proposition}

\begin{proof}Let $G_\theta$ denote the cumulative distribution function (c.d.f.) of $-T$, i.e., $G(x)=\P(-T\leqslant x)$. Denote
\begin{equation}G^{-1}(t)=\inf\{x:\ G(x)>t\}.\label{finv}\end{equation}
For $\theta\in\Theta_0$, we have\begin{eqnarray}&&\P(P<\alpha)=\P\left(\sup_{\phi\in\Theta_0}\P(T_\phi^*\geqslant T\mid T)<\alpha\right)\nonumber\\&&\leqslant
\P\big(\P(T_\theta^*\geqslant T\mid T)<\alpha\big)=\P\big(G(-T)<\alpha\big)\nonumber\\&&\leqslant \P\big(-T<G^{-1}(\alpha)\big)\leqslant\alpha.\label{ieq}\end{eqnarray}
This completes the proof.\end{proof}

Proposition \ref{prop: fp} is a general result, which does not requires any assumption on $T$. From Proposition \ref{prop: fp},  a test is obtained by solving an
stochastic optimization problem in \eqref{pv}, which can be rewritten as \begin{equation}P=\sup_{\phi\in\Theta_0}\int I(T({x})\geqslant
t)dF({x},\phi),\label{so}\end{equation}where $I$ is the indicator function and $t$ is the realization of $T$. In principle, any hypothesis testing problem can be solved
by this way as long as the corresponding optimization problem in \eqref{so} is solvable. In limited trivial cases, the problem in \eqref{so} has obvious solution; an
example is the one-sided $Z$-test. However, except for such cases, this method faces some difficulties in computation: the stochastic optimization problem is generally
very hard to solve, especially when $\Theta_0$ is an unbounded set.

In statistical literature, a commonly used strategy to overcome these difficulties is based on the asymptotic distribution of the test statistic $T$. The optimization
problem in \eqref{so} is often solvable when replacing the distribution of $T$ by its asymptotic distribution. For example, with a $T$ whose asymptotic distribution is
free of unknown parameters, it is trivial to solve \eqref{so}. For complex problems, it is often not easy to derive the asymptotic distribution, or to find such a $T$
whose asymptotic distribution has desirable properties. A Bayesian remedy is Meng (1994)'s posterior predictive $p$-value, which averages the objective function in
\eqref{so} over the posterior distribution of the parameter under the null hypothesis.

Here we provide a more general strategy without any requirement on the distribution of $T$. Suppose that $H_0$ holds. For the true parameter $\theta\in\Theta_0$, it
suffices to obtain a $p$-value that controls Type I error by optimizing the objective function in \eqref{pv} over any set that contains $\theta$, instead of over the
whole $\Theta_0$; see the first inequality in \eqref{ieq}. Consequently, we need to compute\begin{equation}P_0=\max_{\phi\in {\mathcal{N}}(\theta)\cap\Theta_0}\int
I(T({x})\geqslant t)dF({x},\phi),\label{sos}\end{equation}where ${\mathcal{N}}(\theta)$ is a closed neighborhood of $\theta$ containing $\theta$. Here ``$\sup$" in
\eqref{so} is replaced by ``$\max$" if we assume that ${\mathcal{N}}(\theta)\cap\Theta_0$ is a compact subset of $\Theta_0$ on which $\int I(T({x})\geqslant
t)dF({x},\phi)$ is continuous with respect to $\phi$. In practice, we use a consistent estimator $\hat{\theta}$ of $\theta$ under $H_0$ to replace $\theta$ in
\eqref{sos}, and obtain\begin{equation}P_{\mathrm{LOT}}=\max_{\phi\in {\mathcal{N}}(\hat{\theta})\cap\Theta_0}\int I(T({x})\geqslant
t)dF({x},\phi).\label{otp}\end{equation}If the probability of $\theta\in{\mathcal{N}}(\hat{\theta})$ tends to one, then the test based on the $p$-value in \eqref{otp} is
asymptotically valid. We call this test \emph{local optimization-based test} (LOT) throughout the paper. LOT only requires the maximum value of the objective function
over a neighborhood of $\hat{\theta}$, which can be achieved by standard optimization techniques. This feature makes LOT work for many complex problems, in which it is
hard to analyze the distribution of $T$.

When ${\mathcal{N}}(\hat{\theta})$ shrinks to $\hat{\theta}$, \eqref{otp} becomes \begin{equation}P_{\mathrm{B}}=\int I(T({x})\geqslant
t)dF({x},\hat{\theta}),\label{bp}\end{equation}which is the $p$-value of the bootstrap test (Davison and Hinkley 1997). Therefore, LOT can be viewed as an extension of
the bootstrap test. LOT always controls Type I error asymptotically as long as $\hat{\theta}$ is a consistent estimator, whereas the bootstrap test can fail for
non-regular cases where the bootstrap distribution estimator is inconsistent (Bickel and Ren 2001). From \eqref{pv}, \eqref{otp}, to \eqref{bp}, LOT is a bridge
connecting Fisher's significance test and Efron's bootstrap test; see Table \ref{tab:ct}.

\begin{table}[h]\begin{center}\scriptsize\caption{\label{tab:ct} Comparison of three tests} \centering\vspace{3mm}
\begin{tabular}{lcccc}\hline\\[-1.5mm]&\quad&How to lie in Neyman-Pearson's framework&\quad&Difficulty level in implementation\\\hline
\\[-1mm]Fisher's significance test  &&always&&high
\\[1mm]LOT        &&under weak conditions&&moderate
\\[1mm]Efron's bootstrap test&&under strong conditions&& low
\\\hline\end{tabular}\end{center}\end{table}

\section{Local optimization-based interval estimation}\label{sec:ci}
\hskip\parindent \vspace{-0.8cm}


The idea of approximating the $p$-value via local optimization can be modified to construct confidence intervals. Suppose that the parameter of interest is
$\xi=\xi(\theta)\in{\mathbb{R}}$, and that $\hat{\xi}=\hat{\xi}(\X)$ is an estimator of $\xi$. Let $H_\theta$ denote the c.d.f. of the pivotal quantity $\xi-\hat{\xi}$,
i.e., $H_\theta(x)=\P(\xi-\hat{\xi}\leqslant x)$. It should be pointed out that the (asymptotic) distribution of $\xi-\hat{\xi}$ is allowed to depend on unknown
parameters, and this is different from the standard definition of a pivotal quantity in textbooks. Define $H_\theta^{-1}$ as in \eqref{finv}.
\begin{proposition}\label{prop: ci}For all $\theta\in\Theta$ and $\alpha\in(0,1)$,
\begin{eqnarray}&&\P\left(\xi\leqslant\hat{\xi}+\sup_{\phi\in\Theta}H_\phi^{-1}(1-\alpha)\right)
\geqslant 1-\alpha,\label{ul}\\&&\P\left(\xi\geqslant\hat{\xi}+\inf_{\phi\in\Theta}H_\phi^{-1}(\alpha)\right) \geqslant
1-\alpha.\label{ll}\end{eqnarray}\end{proposition}

\begin{proof}We have
\begin{eqnarray}&&\P\left(\xi\leqslant\hat{\xi}+\sup_{\phi\in\Theta}H_\phi^{-1}(1-\alpha)\right)\geqslant
\P\left(\xi-\hat{\xi}\leqslant H_{\theta}^{-1}(1-\alpha)\right)\nonumber\\&&=H_{\theta}\left(H_{\theta}^{-1}(1-\alpha)\right)\geqslant1-\alpha. \label{uls}\end{eqnarray}
This completes the proof of \eqref{ul}, and that of \eqref{ll} is similar. \end{proof}

By Proposition \ref{prop: ci}, the upper and lower $1-\alpha$ confidence bounds of $\xi$ are given by $\hat{\xi}+\sup_{\phi\in\Theta}H_\phi^{-1}(1-\alpha)$ and
$\hat{\xi}+\inf_{\phi\in\Theta}H_\phi^{-1}(\alpha)$, respectively. The equal-tailed $1-\alpha$ confidence interval of $\xi$ is
$\big[\hat{\xi}+\inf_{\phi\in\Theta}H_\phi^{-1}(\alpha/2),\ \hat{\xi}+\sup_{\phi\in\Theta}H_\phi^{-1}(1-\alpha/2)\big]$. These interval limits all need to solve an
optimization problem \begin{equation*}\sup_{\phi\in\Theta}H_\phi^{-1}(\gamma)\quad \text{or}\quad \inf_{\phi\in\Theta}H_\phi^{-1}(\gamma)\end{equation*}for some
$\gamma\in(0,1)$, which is often difficult. Like \eqref{sos}, Proposition \ref{prop: ci} also holds if we take supremum over an arbitrary region containing the true
value of $\theta$; see the first inequality in \eqref{uls}. Suppose that $\hat{\theta}$ is a consistent estimator of $\theta$. Under some mild conditions, we can get
asymptotically valid confidence limits through solving \begin{equation}\sup_{\phi\in\mathcal{N}(\hat{\theta})}H_\phi^{-1}(\gamma)\quad \text{or}\quad
\inf_{\phi\in\mathcal{N}(\hat{\theta})}H_\phi^{-1}(\gamma). \label{cio}\end{equation} Specifically, the upper and lower $1-\alpha$ confidence bounds of $\xi$ are
$\hat{\xi}+\sup_{\phi\in\mathcal{N}(\hat{\theta})}H_\phi^{-1}(1-\alpha)$ and $\hat{\xi}+\inf_{\phi\in\mathcal{N}(\hat{\theta})}H_\phi^{-1}(\alpha)$, respectively, and
the equal-tailed $1-\alpha$ confidence interval of $\xi$ is $\big[\hat{\xi}+\inf_{\phi\in\mathcal{N}(\hat{\theta})}H_\phi^{-1}(\alpha/2),\
\hat{\xi}+\sup_{\phi\in\mathcal{N}(\hat{\theta})}H_\phi^{-1}(1-\alpha/2)\big]$. Here ``$\sup$" (or ``$\inf$") can be replaced by ``$\max$" (or ``$\min$") if
${\mathcal{N}}(\hat{\theta})$ is a compact subset of $\Theta$ on which $H_\phi^{-1}$ is continuous with respect to $\phi$. We call these confidence intervals \emph{local
optimization-based confidence intervals} (LOCIs) throughout the paper. When ${\mathcal{N}}(\hat{\theta})$ shrinks to $\hat{\theta}$, LOCIs become the bootstrap hybrid
confidence intervals (Shao and Tu 1995).

\section{Asymptotic properties}\label{sec:ap}\hskip\parindent \vspace{-0.8cm}

This section discusses asymptotic properties of the proposed local optimization-based methods. Further results involving some computational method are deferred in the
Appendix. Here we only consider one-sided LOCIs, and similar results also hold for two-sided LOCIs and LOTs. Some notation and definitions are needed. The parameter
space $\Theta$ is assumed to be a metric space with metric $\rho$. For $A\subset\Theta$, let $|A|$ denote $\max\{\rho(a,b):\ a,b\in A\}$. For two c.d.f.'s $F_1$ and
$F_2$, the Kolmogorov distance between them is defined as $d_{\mathrm{K}}(F_1,F_2)=\sup_{x\in\mathbb{R}}|F_1(x)-F_2(x)|$. We allow the neighborhood $\mathcal{N}(\cdot)$
to depend on $n$ and denote $\mathcal{N}_n(\cdot)$ for clarity. We use ``$\rightarrow_d$" to denote ``converge in distribution", and let ``a.s." be the abbreviation for
``almost surely". As in Section \ref{sec:ci}, let $H_\theta$ denote the c.d.f. of $\xi-\hat{\xi}$. Since $\mathcal{N}_n(\hat{\theta})$ is a random set, for
$\phi\in\mathcal{N}_n(\hat{\theta})$, $H_\phi$ is actually a random c.d.f., i.e., $H_\phi(x)=\P\big(\xi(\phi)-\hat{\xi}(\X^*)\leqslant x|\,\X\big)$, where the
conditional distribution of $\X^*$ conditional on $\X$ is $F(\cdot,\phi)$.

\begin{assumption}\label{as:n}As $n\to\infty$, $\P\left(\theta\in\mathcal{N}_n(\hat{\theta})\right)\to1$ for all $\theta\in\Theta$.\end{assumption}

If $\hat{\theta}$ is consistent, then $\mathcal{N}_n(\hat{\theta})$ is easy to construct to satisfy Assumption \ref{as:n}; see \eqref{a1} in Section \ref{subsec:N}. We
can immediately have the following theorem.
\begin{theorem}\label{th: acr}Under Assumption \ref{as:n}, for all $\theta\in\Theta$ and $\alpha\in(0,1)$,
\begin{eqnarray}\liminf_{n\to\infty}\P\left(\xi\leqslant\hat{\xi}+\sup_{\phi\in\mathcal{N}_n(\hat{\theta})}H_\phi^{-1}(1-\alpha)\right)
\geqslant 1-\alpha.\label{th1}\end{eqnarray}\end{theorem}

We next show that LOCIs are first order asymptotically equivalent to the bootstrap confidence intervals under regularity conditions. Specifically, if the bootstrap
distribution estimator of $\xi-\hat{\xi}$ is consistent, then ``$\geqslant$" in \eqref{th1} can be replaced by ``$=$". Several assumptions are needed.

\begin{assumption}\label{as:nsz}As $n\to\infty$, $|\mathcal{N}_n(\hat{\theta})|\to0$ (a.s.).\end{assumption}

\begin{assumption}\label{as:sc}As $n\to\infty$, $\hat{\theta}\to\theta$ (a.s.) for all $\theta\in\Theta$.\end{assumption}

\begin{assumption}\label{as:bc}(i) There exists a series of numbers $a_n\to\infty$ such that $a_n(\xi-\hat{\xi})\rightarrow_d K$, where $K$ is
a continuous c.d.f. and is strictly increasing on its support.\\(ii) For $\phi\in\mathcal{N}_n(\hat{\theta})$, $d_{\mathrm{K}}(\bar{H}_{\phi},K)\to0$ (a.s.), where
$\bar{H}_{\phi}(x)=\P\big(a_n[\xi(\phi)-\hat{\xi}(\X^*)]\leqslant x|\,\X\big)$ and $\X^*$ is the bootstrap sample drawn from $F(\cdot,\phi)$.\end{assumption}

Assumption \ref{as:bc} indicates that the bootstrap distribution estimator of $a_n(\xi-\hat{\xi})$ is consistent (Shao and Tu 1995). We can use the conditional
distribution of $a_n[\xi(\tilde{\theta})-\hat{\xi}(\X^*)]$ conditional on $\X$ to approximate that of $a_n(\xi-\hat{\xi})$, and this approximation leads to
asymptotically valid confidence intervals for $\xi$. Assumption \ref{as:bc} holds for general regular cases. We present two simple examples.

\begin{example}\label{ex:bino}Let $X_n$ be a random number from a binomial distribution $\mathrm{BN}(n,\pi)$ with parameter $\pi\in(0,1)$. Consider the pivotal quantity
$\pi-X_n/n$. It is clear that $\sqrt{n}(\pi-X_n/n)\rightarrow_d K$, where $K$ is the c.d.f. of $N(0,\pi(1-\pi))$. This result also holds for any strongly consistent
estimator $\tilde{\pi}_n$ of $\pi$. Specifically, with $X_n^*\sim\mathrm{BN}(n,\tilde{\pi}_n)$, we can easily prove that $d_{\mathrm{K}}(\bar{H}_{\tilde{\pi}},K)\to0$
(a.s.) by the central limit theorem for triangle arrays, where $\bar{H}_{\tilde{\pi}}(x)=P_{{{\pi}}}\big(\sqrt{n}(\tilde{\pi}_n-X^*/n)\leqslant x|\,X_n\big)$, and then
Assumption \ref{as:bc} holds.\end{example}

\begin{example}\label{ex:nonp}Let $X_1,\ldots,X_n$ be i.i.d. random variables from a c.d.f. $F$ with $E X_1^4<\infty$. Here we do not assume
a parametric form for $F$. Then the parameter space $\Theta=\{F\in{\mathcal{F}}:\ \int x^4 dF(x)<\infty\}$ is an infinite-dimensional metric space with metric
$d_{\mathrm{K}}$, where ${\mathcal{F}}$ denotes the set of all c.d.f.'s on $\mathbb{R}$. A strongly consistent estimator of $F$ is the empirical distribution
$\hat{F}(x)=\sum_{i=1}^nI(X_i\leqslant x)/n$. Suppose that the parameter of interest is $\mu=E X_1$. Let $\bar{X}_n$ denote the sample mean. Consider the pivotal
quantity $\mu-\bar{X}_n$. First, we have $\sqrt{n}(\mu-\bar{X}_n)\rightarrow_d \Phi\big(\cdot/v(F)\big)$, where $\Phi$ is the c.d.f. of $N(0,1)$ and
$v(F)=\int\big(x-\int xd F(x)\big)^2 d F(x)$. Second, take
\begin{equation}\label{n}\mathcal{N}_n(\hat{F})=\{G\in\Theta:\ d_{\mathrm{K}}(G,\hat{F})<1/{n}^{1/3},\ |v(G)-v(\hat{F})|<1/{n}^{1/3}\}.\end{equation}
It is easy to verify Assumptions \ref{as:n}-\ref{as:sc}. Furthermore, for $F_n\in\mathcal{N}_n(\hat{F})$ and $X_1^*,\ldots,X_n^*$ i.i.d. from $F_n$, through verifying
the Lindeberg condition in the central limit theorem for triangle arrays, we have that $d_{\mathrm{K}}(\bar{H}_{F_n}^{\mathrm{s}},\Phi\big)\to0$ (a.s.), where
$\bar{H}_{F_n}^{\mathrm{s}}(x)=\P\big(\sqrt{n}(E X_1^* -\bar{X}_n^*)/v(F_n)\leqslant x|\,X_1,\ldots,X_n\big)$. Denote $\bar{H}_{F_n}(x)=\P\big(\sqrt{n}(E X_1^*
-\bar{X}_n^*)\leqslant x|\,X_1,\ldots,X_n\big)$. By \eqref{n}, \\$d_{\mathrm{K}}\left(\bar{H}_{F_n}(\cdot),\Phi\big(\cdot/v(F)\big)\right)\to0$ (a.s.). Then Assumption
\ref{as:bc} holds. \end{example}

\begin{theorem}\label{th:aeb}Under Assumptions \ref{as:n}--\ref{as:bc}, for all $\theta\in\Theta$ and $\alpha\in(0,1)$,
\begin{eqnarray*}\lim_{n\to\infty}\P\left(\xi\leqslant\hat{\xi}+\sup_{\phi\in\mathcal{N}_n(\hat{\theta})}H_\phi^{-1}(1-\alpha)\right)
=1-\alpha.\end{eqnarray*}\end{theorem}

\begin{proof} For any $n$, there exists $\theta_n^*\in\mathcal{N}_n(\hat{\theta})$ such that $\sup_{\phi\in\mathcal{N}_n(\hat{\theta})}\bar{H}_\phi^{-1}(1-\alpha)
<\bar{H}_{\theta_n^*}^{-1}(1-\alpha)+1/n$. Under Assumptions \ref{as:nsz} and \ref{as:sc}, $\theta_n^*\to\theta_0$ (a.s.). Therefore, by Assumption \ref{as:bc},
$\bar{H}_{\theta_n^*}^{-1}(1-\alpha)\to K^{-1}(1-\alpha)$ (a.s.). We have
\begin{eqnarray*}&&\P\left(\xi\leqslant\hat{\xi}+\sup_{\phi\in\mathcal{N}_n(\hat{\theta})}H_\phi^{-1}(1-\alpha)\right)
=\P\left(a_n(\xi-\hat{\xi})\leqslant\sup_{\phi\in\mathcal{N}_n(\hat{\theta})}\bar{H}_\phi^{-1}(1-\alpha)\right)\\&&\leqslant \P\left(a_n(\xi-\hat{\xi})\leqslant
\bar{H}_{\theta_n^*}^{-1}(1-\alpha)+1/n\right)=\P\left(a_n(\xi-\hat{\xi})\leqslant K^{-1}(1-\alpha)+o(1)\right)
\\&&=\bar{H}_\theta\left(K^{-1}(1-\alpha)+o(1)\right)\to 1-\alpha.\end{eqnarray*}Combining this result with Theorem \ref{th: acr}, we complete the proof.\end{proof}

When applying bootstrap to a specific problem, we need to verified Assumptions \ref{as:sc} and \ref{as:bc} to guarantee its frequentist properties. Theorems \ref{th:
acr} and \ref{th:aeb} indicate that we do not need to do such theoretical work when using LOCI. With a proper $\mathcal{N}_n(\hat{\theta})$, LOCI possesses both the
basic frequentist property in \eqref{th1} and a potential bonus: it enjoys the same first order frequentist property as the bootstrap method when the two assumptions
hold (although we may not know this). It can be expected that, under much stronger conditions, LOCI has some high-order asymptotic properties like bootstrap (Hall 1992).
We do not discuss this here since it is difficult to specify $\mathcal{N}_n(\hat{\theta})$ satisfying such conditions for complex problems.

\section{Implementation}\label{sec:imp}\hskip\parindent \vspace{-0.8cm}

This section discusses how to implement LOT and LOCI. We focus on the cases where $\Theta$ is a subset of an Euclidean space. Therefore, it suffices to solve
finite-dimensional optimization problems in LOT and LOCI. For some problems with infinite-dimensional parameter spaces, LOT or LOCI is still available through rational
simplification; see Section \ref{subsec:nr}.

\subsection{Specification of ${\mathcal{N}}(\hat{\theta})$}\label{subsec:N} \hskip\parindent \vspace{-0.8cm}

The first issue is to determine the neighborhood ${\mathcal{N}}(\hat{\theta})$ in \eqref{otp} and \eqref{cio} over which we solve the optimization problem. Suppose that
the dimension of $\Theta$ is $q$ and $\hat{\theta}=(\hat{\theta}_1,\ldots,\hat{\theta}_q)'$ is a consistent estimator of $\theta=(\theta_1,\ldots,\theta_q)'$. The basic
principle is to select ${\mathcal{N}}(\hat{\theta})$ satisfying Assumption \ref{as:n}. A simple choice of ${\mathcal{N}}(\hat{\theta})$
is\begin{equation}\big[\hat{\theta}_1-\delta,\ \hat{\theta}_1+\delta\big]\times \cdots\times\big[\hat{\theta}_{q}-\delta,\
\hat{\theta}_{q}+\delta\big]\label{a1}\end{equation}for some small constant $\delta>0$. If we know further the convergence rate of $\hat{\theta}$, then the second
principle is to select ${\mathcal{N}}(\hat{\theta})$ satisfying Assumption \ref{as:nsz}. By Theorem \ref{th:aeb}, this selection can make the local optimization-based
method asymptotically equivalent to bootstrap if the bootstrap distribution estimator is consistent. For example, with $\|\hat{\theta}-\theta\|=O_p(1/\sqrt{n})$, a
selection of ${\mathcal{N}}(\hat{\theta})$ simultaneously satisfying Assumptions \ref{as:n} and \ref{as:nsz} is
\begin{equation}\big[\hat{\theta}_1-\delta\log(n)/\sqrt{n},\ \hat{\theta}_1+\delta\log(n)/\sqrt{n}\big]\times \cdots\times\big[\hat{\theta}_{q}-\delta\log(n)/\sqrt{n},\
\hat{\theta}_{q}+\delta\log(n)/\sqrt{n}\big]\label{a12}\end{equation}for some constant $\delta>0$. The constant $\delta$ in \eqref{a1} or \eqref{a12} can be specified
empirically. For complex problems, the convergence rate of $\hat{\theta}$ is difficult to exactly know. We will see in Section \ref{sec:exm} that, LOT or LOCI has good
finite-sample performance even with a simple ${\mathcal{N}}(\hat{\theta})$ like in \eqref{a1} that only satisfies Assumption \ref{as:n}.

It seems more reasonable if the variances of $\hat{\theta}_j$'s are used to construct ${\mathcal{N}}(\hat{\theta})$. When the variance estimators are not
straightforward, the jackknife, bootstrap (Shao and Tu 1995), or even Bayesian methods can be used to estimate the variances. However, such methods will add extra
theoretical and computational work, and there are still some constants, which need to be specified empirically, in the final form of ${\mathcal{N}}(\hat{\theta})$.
Therefore, we suggest using the variance estimators only when they are straightforward.

\subsection{Importance sampling-based approach}\label{subsec:is} \hskip\parindent \vspace{-0.8cm}

Suppose that $F(\cdot,\theta)$ has a probability density function (p.d.f.) $f(\cdot,\theta)$ with respect to a $\sigma$-finite measure $\nu$, and that
$\{f(\cdot,\theta): \phi\in\Theta_0\}$ has a common support. We use an \emph{importance sampling-based approach} to solve the stochastic optimization problem in
\eqref{otp}. First we approximate the objective function in \eqref{otp} by importance sampling. Note that
$$u(\phi)=\int I(T(\v{x})\geqslant t)\frac{f(\v{x},\phi)}{f(\v{x},\hat{\theta})} f(\v{x},\hat{\theta})d\,\nu(\v{x})=
E\left\{I(T(\X^*)\geqslant t)\frac{f(\X^*,\phi)}{f(\X^*,\hat{\theta})} \right\},$$where $\X^*\sim f(\cdot,\hat{\theta})$. According to the sample averaging approximation
method in stochastic optimization (Shapiro 2003), we compute the $p$-value as\begin{equation}P_{\mathrm{IS}}=\max_{\phi\in
{\mathcal{N}}(\hat{\theta})\cap\Theta_0}\hat{u}(\phi),\label{otpa}\end{equation}where
\begin{equation}\hat{u}(\phi)=\frac{1}{M}\sum_{m=1}^M\left\{I(T(\X_m^*)\geqslant t)\frac{f(\X_m^*,\phi)}{f(\X_m^*,\hat{\theta})} \right\},\label{ra}
\end{equation}is the approximation of $r(\theta)$ based on $\X_1^*,\ldots,\X_M^*$ i.i.d.
from $f(\cdot,\hat{\theta})$ with the Monte Carlo sample size $M$. With sufficiently large $M$, $P_{\mathrm{IS}}$ can be arbitrarily close to $P$ in \eqref{otp}. There
are many available iterative algorithms for solving the deterministic optimization problem in \eqref{otpa} such as the interior point method (Boyd and Vandenberghe
2004).

We can also use an experimental design-based method to approximate the $p$-value in \eqref{otpa}. Take $L$ points $\phi_1,\ldots,\phi_L$ uniformly spaced over
${\mathcal{N}}(\hat{\theta})\cap\Theta_0$, and then compute
\begin{equation}P_{\mathrm{IS-D}}=\max\left\{\hat{u}(\hat{\theta}),\hat{u}(\phi_1),\ldots,\hat{u}(\phi_L)\right\},\label{nbpa}\end{equation}
where $\hat{r}$ is defined in \eqref{ra}. We call these points \emph{try points} throughout this paper, which can be constructed from so-called \emph{space-filling
designs} in experimental design; see Section \ref{subsec:d}. Since ${\mathcal{N}}(\hat{\theta})$ is a small neighborhood, $P_{\mathrm{IS-D}}$ often performs well with a
moderate $L$. The design-based method is very easy to implement, and is suitable for those who are not familiar with optimization methods. More sophisticated
space-filling design-based optimization method can be found in Fang, Hickernell, and Winker (1996).

For LOCI, we have the following importance sampling-based method to compute the interval limits when $F(\cdot,\theta)$ has a p.d.f. $f(\cdot,\theta)$ and
$\{f(\cdot,\theta): \theta\in\Theta\}$ has a common support. Here we only consider the computation of upper limits, i.e., the first optimization problem in \eqref{cio}.
Let $\varphi=H_\phi^{-1}(\gamma)$ and $S(\phi,\varphi)=H_\phi(\varphi)$. Suppose that $H_\phi$ is continuous and strictly increasing on its support for
$\phi\in\mathcal{N}(\hat{\theta})$. The problem \eqref{cio} is equivalent to the constrained optimization problem
\begin{equation}\max_{\phi\in\mathcal{N}(\hat{\theta})}\varphi \quad\text{subject to}\ S(\phi,\varphi)=\gamma.\label{cio2}\end{equation}
For $q$-dimensional space $\Theta$, the problem optimizes $q+1$ variables. Similar to the importance sampling-based sample averaging approximation method in \eqref{ra},
we use an approximation of $S$,
$$\hat{S}(\phi,\varphi)=\frac{1}{M}\sum_{m=1}^M\left\{I\big(\xi(\phi)-\hat{\xi}(\X_m^*)\leqslant \varphi\big)\frac{f(\X_m^*,\phi)}{f(\X_m^*,\hat{\theta})} \right\},$$
where $\X_1^*,\ldots,\X_M^*$ are i.i.d. from $f(\cdot,\hat{\theta})$ with the Monte Carlo sample size $M$. The solution
to\begin{equation}\max_{\phi\in\mathcal{N}(\hat{\theta})}\varphi \quad\text{subject to}\ \hat{S}(\phi,\varphi)=\gamma\label{cioa}\end{equation}can be used to approximate
that to \eqref{cio2}. Note that $\hat{S}(\phi,\varphi)$ may not equal $\gamma$ exactly in \eqref{cioa}. In practice we handle an equivalent problem
\begin{equation}\max_{\phi\in\mathcal{N}(\hat{\theta})}\varphi \quad\text{subject to}\ \hat{S}(\phi,\varphi)\leqslant \gamma \label{cioa2}\end{equation}
instead of \eqref{cioa}. A design-based method similar to \eqref{nbpa} can also be used to solve \eqref{cioa2}. Since \eqref{cioa2} has not straightforward solution even
for a given $\phi\in\mathcal{N}(\hat{\theta})$, we do not recommend such a method. A more simple and general method for computing LOCIs is to directly compute the
quantiles of $H_\phi$ for a given $\phi$. This method will be discussed in the next subsection.

\subsection{Neighborhood bootstrap}\label{subsec:nb} \hskip\parindent \vspace{-0.8cm}

This subsection discusses a general method, called \emph{neighborhood bootstrap}, to implement LOT and LOCI. This method still works for the cases where the importance
sampling-based approach in Section \ref{subsec:is} fails. We first consider LOT. Like the design-based $p$-value in \eqref{nbpa}, take $L$ try points
$\phi_1,\ldots,\phi_L$ uniformly spaced over ${\mathcal{N}}(\hat{\theta})\cap\Theta_0$. The difference from \eqref{nbpa} is that the neighborhood bootstrap method
directly approximates the objective value in \eqref{otp} by the Monte Carlo method. Specifically, for each $\phi_l,\ l=0,1,\ldots,L$, generate
$\X_{l,1}^*,\ldots,\X_{l,M}^*$ i.i.d. from $F(\cdot,\phi_l)$, where $\phi_0=\hat{\theta}$. Then the $p$-value in \eqref{otp} can be approximated by
\begin{equation*}P_{\mathrm{NB}}=\max_{l=0,\ldots,L}\left\{\frac{1}{M}\sum_{m=1}^MI\big(T(\X_{l,m}^*)\geqslant t\big)\right\}. \end{equation*}

For LOCI, we still consider the computation of upper limits in \eqref{cio}. With $\{\phi_1,\ldots,\phi_L\}$ uniformly spaced over ${\mathcal{N}}(\hat{\theta})$, take
bootstrap sample  $\X_{l,1}^*,\ldots,\X_{l,M}^*$ i.i.d. from $F(\cdot,\phi_l)$ for $l=0,1,\ldots,L$. Let $\widehat{H_{\phi_l}^{-1}(\gamma)}$ denote the sample
$\gamma$-quantile of $\xi(\phi_l)-\hat{\xi}(\X_{l,1}^*),\ldots,\xi(\phi_l)-\hat{\xi}(\X_{l,M}^*)$. Consequently,
$\sup_{\phi\in\mathcal{N}(\hat{\theta})}H_\phi^{-1}(\gamma)$ can be approximated by
\begin{equation}\max_{l=0,\ldots,L}\left\{\widehat{H_{\phi_l}^{-1}(\gamma)}\right\}.\label{cinb}\end{equation}

Neighborhood bootstrap is a very general method. In principle, it can be applied to infinite-dimensional parameter spaces if there are well-defined space-filling designs
for such spaces. Another advantage of neighborhood bootstrap is its easy implement, especially for computing LOCIs. For LOT, neighborhood bootstrap is slightly more
time-consuming than the importance sampling-based approach.

\subsection{Design of try points}\label{subsec:d} \hskip\parindent \vspace{-0.8cm}

The design-based $p$-value in \eqref{nbpa} and the neighborhood bootstrap method in Section \ref{subsec:nb} both need $L$ try points $\phi_1,\ldots,\phi_L$ uniformly
spaced over ${\mathcal{N}}(\hat{\theta})$. This subsection presents some discussion on the design of these points. Usually ${\mathcal{N}}(\hat{\theta})$ is selected as a
$q$-dimensional hypercube like \eqref{a1} or \eqref{a12}. Specifically, suppose that ${\mathcal{N}}(\hat{\theta})=[L_1,U_1]\times\cdots\times[L_q,U_q]$. For
$\psi_i=(\psi_{i1},\ldots,\psi_{iq})'\in[0,1]^q,\ i=1,\ldots,L$, let $\phi_{ij}=L_j+\psi_{ij}(U_j-L_j),\ i=1,\ldots,L,\ j=1,\ldots,q$, and we have
$\phi_i=(\phi_{i1},\ldots,\phi_{iq})'\in{\mathcal{N}}(\hat{\theta})$ for $i=1,\ldots,L$. Therefore, it suffices to consider the design of $\psi_1,\ldots,\psi_L$ in
$[0,1]^q$, called \emph{initial design} in the following. As mentioned in Section \ref{subsec:is}, the initial design can be constructed from space-filling designs in
$[0,1]^q$. Such designs include grids, Latin hypercube designs (McKay, Beckman, and Conover 1979), and uniform designs (Fang et al. 2000), among others. A simple choice
is the following grid
\begin{equation}\left\{\frac{1}{2U},\ldots,\frac{2U-1}{2U}\right\}\times\cdots\times\left\{\frac{1}{2U},\ldots,\frac{2U-1}{2U}\right\},\label{grid}\end{equation} where
$U$ is a positive integer. There are $L=U^q$ points in the grid, and this leads to unaffordable computations for large $q$. Another choice is the Latin hypercube design
(LHD) (McKay, Beckman, and Conover 1979), which is easy to construct for any $L$ and $q$. The LHD is spaced uniformly in each dimension, and its space-filling properties
over the whole $[0,1]^q$ can be improved by iterative algorithms (Park 2001). There are functions for generating LHDs in both \texttt{MATLAB} and \texttt{R}.

Note that in fact we need to design $\phi_1,\ldots,\phi_L$ in ${\mathcal{N}}(\hat{\theta})\cap\Theta_0$ for LOT or in ${\mathcal{N}}(\hat{\theta})\cap\Theta$ for LOCI.
For irregular or constrained parameter spaces, this problem becomes complicated. A feasible solution is to design more points in ${\mathcal{N}}(\hat{\theta})$ and then
to keep those in the intersection.

\section{Illustrative examples}\label{sec:exm}\hskip\parindent \vspace{-0.8cm}

This section presents four examples to illustrate LOT and LOCI, in which the (asymptotic) distributions of the test statistics or pivotal quantities are non-regular or
unclear.

\subsection{Interval estimation for the maximum cell probability of the multinomial distribution}\label{subsec:pmax} \hskip\parindent \vspace{-0.8cm}

Let $(X_{n1},\ldots,X_{nk})'$ be the cell frequencies from a multinomial distribution, $\mathrm{MN}_k(n;\pi)$, where $\sum_{i=1}^kX_{ni}=n,$ with the parameter
$\pi=(\pi_1,\ldots,\pi_k)',\ \pi_i>0,\ i=1,\ldots,k$, and $\sum_{i=1}^k\pi_i=1$. We consider interval estimation for $\pi_{\max}=\max\{\pi_1,\ldots,\pi_k\}$. This
problem is related to some real applications including the diversity of ecological populations (Patil and Taillie 1979) and favorable numbers on a roulette wheel (Ethier
1982), and has been studied by Gelfand et al. (1992), Glaz and Sison (1999), and Xiong and Li (2009), among others.

The maximum likelihood estimator (MLE) of $\pi$ is $(X_{n1}/n,\ldots,X_{nk}/n)'$, and that of $\pi_{\max}$ is $\max_{1\leqslant i\leqslant k}X_{ni}/n$. To avoid extreme
values in the estimators, we use the Bayesian estimator $\hat{\pi}=(\hat{\pi}_1,\ldots,\hat{\pi}_k)'=\big((X_{n1}+1/2)/(n+k/2),\ldots,(X_{nk}+1/2)/(n+k/2)\big)'$ from
the Jeffrey prior (Ghosh, Delampady, and Samanta 2007),. The corresponding estimator of $\pi_{\max}$ is $\hat{\pi}_{\max}=\max_{1\leqslant i\leqslant k}\hat{\pi}_i$,
whose asymptotic properties are the same as the MLE. Xiong and Li (2009) showed that, when the numbers in $\{i=1,\ldots,k:\ \pi_i=\pi_{\max}\}$ are more than one,
$\hat{\pi}_{\max}$ is not asymptotically normal and the corresponding bootstrap distribution estimator is inconsistent. A remedy is to use $m$-out-of-$n$ bootstrap
(Bickel, G\"{o}tze, and van Zwet 1997). This method takes bootstrap sample $(X_{m1}^*,\ldots,X_{mk}^*)'$ from $\mathrm{MN}_k(m;\hat{\pi})$ with $m=o(n)$, and then
approximates the distribution of $\sqrt{n}(\hat{\pi}_{\max}-\pi_{\max})$ by its bootstrap analogue $\sqrt{m}(\hat{\pi}_{\max}^*-\hat{\pi}_{\max})$, where
$\hat{\pi}_{\max}^*=\max_{1\leqslant i\leqslant k}\big\{(X_{mi}^*+1/2)/(m+k/2)\big\}$. Xiong and Li (2010) proved that this approximation is consistent, and thus results
in asymptotically valid confidence intervals for $\pi_{\max}$.

The LOCI of $\pi_{\max}$ can be easily constructed by the neighborhood bootstrap method in \eqref{cinb}, where the pivotal quantity is $\pi_{\max}-\hat{\pi}_{\max}$. We
next conduct a simulation study to compare the LOCI with the ordinary bootstrap and $m$-out-of-$n$ bootstrap methods. Here we focus on two-sided $1-\alpha$ confidence
intervals with $\alpha=0.05$. In our simulation study, $k$ is fixed as $5$, and $n=30$ and 60 are considered. We use six vectors of cell probabilities; see Table
\ref{tab:mul}. In the $m$-out-of-$n$ bootstrap method, $m$ is set as the integer part of $2\sqrt{n}$. The neighborhood ${\mathcal{N}}(\hat{\pi})$ is
$$\big[\hat{\pi}_1-\delta\log(n)/\sqrt{n},\ \hat{\pi}_1+\delta\log(n)/\sqrt{n}\big]\times\cdots\times
\big[\hat{\pi}_k-\delta\log(n)/\sqrt{n},\ \hat{\pi}_k+\delta\log(n)/\sqrt{n}\big]$$where two values, $0.1$ and $0.5$, of $\delta$ are used. it is clear that
${\mathcal{N}}(\hat{\pi})$ satisfies Assumptions \ref{as:n} and \ref{as:nsz}. We use two grids in \eqref{grid} to design the try points with $U=3$ for $\delta=0.1$ and
$U=5$ for $\delta=0.5$. Note that there is a constraint $\sum_{i=1}^k\pi_i=1$ in the parameter space. There are 51 and 101 try points in the two grids, respectively. The
bootstrap sample size is 5000 in all the above methods.

{\begin{table}[htbp]\begin{center}\scriptsize\caption{\label{tab:mul}Simulation results in Section \ref{subsec:pmax}} \centering\vspace{3mm}
\begin{tabular}{llccccccc} \multicolumn{9}{c}{$\pi=(0.7, 0.075, 0.075, 0.075, 0.075)'$}\\[1mm]\hline &\quad\quad&\multicolumn{3}{c}{$n=30$}&\quad&\multicolumn{3}{c}{$n=60$}
\\[2mm]&&CR&ML&SDL&&CR&ML&SDL\\\hline
\\[-2mm]Bootstrap           &&0.927&0.299&0.023&&0.932&0.222&0.014
\\[1mm]Bootstrap ($m<n$)    &&0.920&0.389&0.054&&0.945&0.377&0.037
\\[1mm]LOCI ($\delta=0.1$)  &&0.950&0.325&0.021&&0.940&0.228&0.013
\\[1mm]LOCI ($\delta=0.5$)  &&0.961&0.345&0.020&&0.954&0.236&0.012
\\\hline \end{tabular}\vspace{3mm}
\begin{tabular}{llccccccc} \multicolumn{9}{c}{$\pi=(0.5, 0.15, 0.15, 0.1, 0.1)'$}\\[1mm]\hline &\quad\quad&\multicolumn{3}{c}{$n=30$}&\quad&\multicolumn{3}{c}{$n=60$}
\\[2mm]&&CR&ML&SDL&&CR&ML&SDL\\\hline
\\[-2mm]Bootstrap           &&0.846&0.297&0.040&&0.912&0.237&0.011
\\[1mm]Bootstrap ($m<n$)    &&0.881&0.368&0.065&&0.955&0.362&0.038
\\[1mm]LOCI ($\delta=0.1$)  &&0.897&0.321&0.033&&0.931&0.244&0.008
\\[1mm]LOCI ($\delta=0.5$)  &&0.967&0.350&0.018&&0.967&0.259&0.008
\\\hline \end{tabular}\vspace{3mm}
\begin{tabular}{llccccccc} \multicolumn{9}{c}{$\pi=(0.3, 0.175, 0.175, 0.175, 0.175)'$}\\[1mm]\hline &\quad\quad&\multicolumn{3}{c}{$n=30$}&\quad&\multicolumn{3}{c}{$n=60$}
\\[2mm]&&CR&ML&SDL&&CR&ML&SDL\\\hline
\\[-2mm]Bootstrap           &&0.738&0.175&0.075&&0.702&0.147&0.058
\\[1mm]Bootstrap ($m<n$)    &&0.748&0.187&0.093&&0.722&0.164&0.090
\\[1mm]LOCI ($\delta=0.1$)  &&0.939&0.210&0.075&&0.832&0.172&0.054
\\[1mm]LOCI ($\delta=0.5$)  &&0.991&0.296&0.046&&0.990&0.217&0.032
\\\hline \end{tabular}\vspace{3mm}
\begin{tabular}{llccccccc} \multicolumn{9}{c}{$\pi=(0.3, 0.3, 0.2, 0.1, 0.1)'$}\\[1mm]\hline &\quad\quad&\multicolumn{3}{c}{$n=30$}&\quad&\multicolumn{3}{c}{$n=60$}
\\[2mm]&&CR&ML&SDL&&CR&ML&SDL\\\hline
\\[-2mm]Bootstrap           &&0.893&0.207&0.064&&0.909&0.168&0.035
\\[1mm]Bootstrap ($m<n$)    &&0.913&0.234&0.081&&0.955&0.210&0.065
\\[1mm]LOCI ($\delta=0.1$)  &&0.954&0.248&0.064&&0.944&0.205&0.036
\\[1mm]LOCI ($\delta=0.5$)  &&0.979&0.327&0.036&&0.966&0.241&0.023
\\\hline \end{tabular}\vspace{3mm}
\begin{tabular}{llccccccc} \multicolumn{9}{c}{$\pi=(0.24, 0.24, 0.24, 0.24, 0.04)'$}\\[1mm]\hline &\quad\quad&\multicolumn{3}{c}{$n=30$}&\quad&\multicolumn{3}{c}{$n=60$}
\\[2mm]&&CR&ML&SDL&&CR&ML&SDL\\\hline
\\[-2mm]Bootstrap           &&0.935&0.170&0.065&&0.924&0.134&0.041
\\[1mm]Bootstrap ($m<n$)    &&0.956&0.184&0.074&&0.986&0.149&0.057
\\[1mm]LOCI ($\delta=0.1$)  &&0.943&0.210&0.064&&0.949&0.174&0.042
\\[1mm]LOCI ($\delta=0.5$)  &&0.976&0.305&0.032&&0.970&0.220&0.024
\\\hline \end{tabular}\vspace{3mm}
\begin{tabular}{llccccccc} \multicolumn{9}{c}{$\pi=(0.2, 0.2, 0.2, 0.2, 0.2)'$}\\[1mm]\hline &\quad\quad&\multicolumn{3}{c}{$n=30$}&\quad&\multicolumn{3}{c}{$n=60$}
\\[2mm]&&CR&ML&SDL&&CR&ML&SDL\\\hline
\\[-2mm]Bootstrap           &&0.906&0.135&0.063&&0.946&0.095&0.049
\\[1mm]Bootstrap ($m<n$)    &&0.982&0.136&0.074&&0.996&0.088&0.059
\\[1mm]LOCI ($\delta=0.1$)  &&0.950&0.175&0.064&&0.963&0.127&0.047
\\[1mm]LOCI ($\delta=0.5$)  &&0.937&0.280&0.038&&0.963&0.195&0.028
\\\hline \end{tabular}
\end{center}\end{table}}

We repeat 5000 times to compute the coverage rates (CRs), mean lengths (MLs), and standard deviations of lengths (SDLs) of the confidence intervals. The simulation
results are shown in Table \ref{tab:mul}. We can see that the bootstrap interval usually has low CR. For dispersed $\pi$, the $m$-out-of-$n$ bootstrap method lacks
efficiency with longer ML, whereas two LOCIs perform better. As expected, the LOCI with $\delta=0.5$ is more conservative than that with $\delta=0.1$. In summary, it can
be concluded that the LOCI is at least comparable to the $m$-out-of-$n$ bootstrap interval.

\subsection{Interval estimation for the location parameter of the three-parameter Weibull distribution}\label{subsec:weibull} \hskip\parindent \vspace{-0.8cm}

The Weibull distribution is widely used in many fields such as survival analysis (Cox and Oakes 1984) and reliability (Murthy, Xie, and Jiang 2004). Let $X_1,\ldots,X_n$
be i.i.d. observations from the Weibull distribution $\mathrm{Wbl}(a,b,\tau)$, whose p.d.f. is
\begin{equation}f(x;a,b,\tau)=\frac{b}{a}\left(\frac{x-\tau}{a}\right)^{b-1}\exp\left[-\left(\frac{x-\tau}{a}\right)^b\right]\label{wbl}\end{equation}
for $x>\tau$, $a>0$, $b>0$, and $\tau\in{\mathbb{R}}$. The parameters $a$, $b$, and $\tau$ are known as the scale, shape, and location parameters, respectively. If
$\tau$ is known, then the likelihood-based inference for the parameters is straightforward (Murthy, Xie, and Jiang 2004). With an unknown $\tau$, the standard method
faces difficulties since the distributions have not a common support (Blischke 1974). Estimation for the parameters of the three-parameter Weibull distribution is still
an active topic in recent years, and many estimators have been proposed; see Lockhart and Stephens (1994), Cousineau (2009), and Teimouri, Hoseini, and Nadarajah (2013),
among others. Since the (asymptotic) distributions of these estimators are difficult to derive, there is limited results on interval estimation for the parameters.

{\begin{table}[htbp]\begin{center}\scriptsize\caption{\label{tab:wei}Simulation results in Section \ref{subsec:weibull}} \centering\vspace{3mm}
\begin{tabular}{llccccccc} \multicolumn{9}{c}{$a=0.5,\ b=0.5$}\\[1mm]\hline &\quad\quad&\multicolumn{3}{c}{$n=10$}&\quad&\multicolumn{3}{c}{$n=20$}
\\[2mm]&&CR&ML&SDL&&CR&ML&SDL\\\hline
\\[-2mm]Bootstrap  &&0.970&0.112&0.075&&0.945&0.023&0.017
\\[1mm]LOCI        &&0.979&0.129&0.190&&0.946&0.024&0.020
\\\hline \end{tabular}\vspace{3mm}
\begin{tabular}{llccccccc} \multicolumn{9}{c}{$a=2.5,\ b=0.5$}\\[1mm]\hline &\quad\quad&\multicolumn{3}{c}{$n=10$}&\quad&\multicolumn{3}{c}{$n=20$}
\\[2mm]&&CR&ML&SDL&&CR&ML&SDL\\\hline
\\[-2mm]Bootstrap  &&0.946&0.336&0.223&&0.940&0.072&0.057
\\[1mm]LOCI        &&0.946&0.343&0.234&&0.940&0.073&0.060
\\\hline \end{tabular}\vspace{3mm}
\begin{tabular}{llccccccc} \multicolumn{9}{c}{$a=0.5,\ b=1.5$}\\[1mm]\hline &\quad\quad&\multicolumn{3}{c}{$n=10$}&\quad&\multicolumn{3}{c}{$n=20$}
\\[2mm]&&CR&ML&SDL&&CR&ML&SDL\\\hline
\\[-2mm]Bootstrap  &&0.410&0.072&0.042&&0.230&0.024&0.012
\\[1mm]LOCI        &&0.949&1.013&0.546&&0.982&0.789&0.316
\\\hline \end{tabular}\vspace{3mm}
\begin{tabular}{llccccccc} \multicolumn{9}{c}{$a=2.5,\ b=1.5$}\\[1mm]\hline &\quad\quad&\multicolumn{3}{c}{$n=10$}&\quad&\multicolumn{3}{c}{$n=20$}
\\[2mm]&&CR&ML&SDL&&CR&ML&SDL\\\hline
\\[-2mm]Bootstrap  &&0.210&0.240&0.133&&0.098&0.055&0.029
\\[1mm]LOCI        &&0.876&1.081&0.644&&0.921&1.166&0.542
\\\hline \end{tabular}\vspace{3mm}
\begin{tabular}{llccccccc} \multicolumn{9}{c}{$a=0.5,\ b=2.5$}\\[1mm]\hline &\quad\quad&\multicolumn{3}{c}{$n=10$}&\quad&\multicolumn{3}{c}{$n=20$}
\\[2mm]&&CR&ML&SDL&&CR&ML&SDL\\\hline
\\[-2mm]Bootstrap  &&0.289&0.104&0.042&&0.121&0.038&0.014
\\[1mm]LOCI        &&0.958&1.454&0.472&&0.972&1.246&0.189
\\\hline \end{tabular}\vspace{3mm}
\begin{tabular}{llccccccc} \multicolumn{9}{c}{$a=2.5,\ b=2.5$}\\[1mm]\hline &\quad\quad&\multicolumn{3}{c}{$n=10$}&\quad&\multicolumn{3}{c}{$n=20$}
\\[2mm]&&CR&ML&SDL&&CR&ML&SDL\\\hline
\\[-2mm]Bootstrap  &&0.014&0.221&0.068&&0.010&0.038&0.020
\\[1mm]LOCI        &&0.881&1.392&0.671&&0.950&1.614&0.379
\\\hline \end{tabular}\end{center}\end{table}}

This subsection constructs LOCIs for $\tau$ based on the maximum product of spacings (MPS) estimation (Cheng and Amin 1983). Obviously our method is also applicable for
other parameters. The MPS estimators $\hat{a}$, $\hat{b}$, and $\hat{\tau}$ are constructed by maximizing
$$S(a,b,\tau)=\prod_{i=1}^{n+1}\int_{X_{(i-1)}}^{X_{(i)}}f(x;a,b,\tau)dx,$$where $X_{(1)}\leqslant\cdots\leqslant X_{(n)}$ are order statistics, $X_{(0)}=\tau$, and
$X_{(n+1)}=\infty$. For all $a,b$, and $\tau$, the MPS estimators are consistent (Cheng and Amin 1983). Furthermore, for $b>2$, they have the same asymptotic
distributions as the MLEs; if $0<b<2$, then $\hat{a}-a=O_p(1/\sqrt{n})$, $\hat{b}-b=O_p(1/\sqrt{n})$, and $\hat{\tau}-\tau=O_p(1/n^{1/b})$. It is not straightforward to
construct confidence intervals of $\tau$ by the asymptotic properties of $\hat{\tau}$ since $b$ is unknown. Furthermore, the validity of the corresponding bootstrap
confidence interval is unclear.

We use neighborhood bootstrap to construct two-sided $1-\alpha$ confidence intervals of $\tau$, and conduct a simulation study to evaluate their performance. The pivotal
quantity is $\tau-\hat{\tau}$. The initial design is the grid in \eqref{grid} with $U=3$ that corresponds to $L=27$. Since the results are sensitive to the value of $b$,
we set the neighborhood ${\mathcal{N}}(\hat{a},\hat{b},\hat{\tau})$ as
$$\big[\hat{a}-\delta_n,\ \hat{a}+\delta_n\big]\times\big[\hat{b}-\delta_n,\ \hat{b}+\delta_n\big]
\times\big[\hat{\tau}-\delta_n,\ \hat{\tau}+\delta_n\big],$$where $\delta_n=4\exp\big(-(1/\hat{b})^5\big)\log(n)/\sqrt{n}$. It is clear that
${\mathcal{N}}(\hat{a},\hat{b},\hat{\tau})$ satisfies Assumptions \ref{as:n} and \ref{as:nsz} for all $a,\ b$, and $\tau$ by the asymptotic properties of the MPSs. For
$\tau=1$, two values of $n$, and several combinations of $(a,b)$, the simulation results based on 1000 repetitions are reported in Table \ref{tab:wei} with
$\alpha=0.05$.  The bootstrap sample sizes used in the bootstrap interval and LOCI are both 1000. We can see that, for $b=0.5$, the CR of the bootstrap interval is
satisfactory, and the LOCI has similar performance to it with slightly longer ML. For larger $b$, the bootstrap interval performs poorly, and the LOCI is much better in
terms of CR.

\subsection{Testing whether all the coefficients in the high-dimensional regression are nonnegative}\label{subsec:nnlasso} \hskip\parindent \vspace{-0.8cm}

High-dimensional data analysis that deals with models where the number of parameters is larger than the sample size is a very active research area in recent years, We
consider the regression model\begin{equation}\label{lm}y=X\beta+\varepsilon,\end{equation} where ${X}=(x_{ij})$ is the $n\times p$ regression matrix,
${y}=(y_1,\ldots,y_n)'$ is the response vector, ${\beta}=(\beta_1,\ldots,\beta_p)'$ is the vector of regression coefficients and
${\varepsilon}=(\varepsilon_1,\ldots,\varepsilon_n)'$ is a vector of i.i.d. normal random errors with zero mean and finite variance $\sigma^2$. Let $p_0$ denote the
number in $\{j=1,\ldots,p:\ \beta_j\neq0\}$. For $p\gg n$, we make the sparsity assumption of $p_0\ll n$. Many methods have been proposed to estimate the sparse $\beta$
in \eqref{lm} such as the lasso (Tibshirani 1996), the smoothly clipped absolute deviation method (Fan and Li 2001), and the minimax concave penalty method (Zhang 2010).
Under the assumption that all the coefficients are known to be nonnegative, Efron et al. (2004) introduced a nonnegative lasso method to estimate $\beta$, which solves
\begin{equation}\min_\beta\|y-X\beta\|^2+\lambda\sum_{j=1}^p\beta_j\quad\text{subject to}\ \beta_j\geqslant0,\ j=1,\ldots,p,\label{nnlasso}\end{equation}
where $\lambda>0$ is a tuning parameter. Applications of this method can be found in Frank and Heiser (2006) and Wu, Yang, and Liu (2014). In this subsection we use the
data to test whether the assumption in the nonnegative lasso method is reasonable, i.e., test the following hypotheses
\begin{equation}H_0:\ \beta_j\geqslant0,\ j=1,\ldots,p\ \leftrightarrow\ H_1:\ H_0\ \text{does not hold}.\label{H0}\end{equation}

In classical $n>p$ settings, the problem to test \eqref{H0} has been discussed by the likelihood ratio test; see Silvapulle and Sen (2011). However, this method cannot
be dirrectly extended to the high-dimensional case since the MLEs perform very poorly for such a case. Here we borrow the idea of the generalized likelihood ratio test
in nonparametric statistics (Fan, Zhang, and Zhang 2001), and construct the test
statistic\begin{equation*}T=\frac{\|y-X\hat{\beta}_{H_0}\|^2}{\|y-X\hat{\beta}_{H_1}\|^2},\end{equation*}where $\hat{\beta}_{H_0}$ and $\hat{\beta}_{H_1}$ are the
estimators of $\beta$ under $H_0$ and $H_1$, respectively. A natural choice is to use the nonnegative lasso estimator in \eqref{nnlasso} and the lasso estimator as
$\hat{\beta}_{H_0}$ and $\hat{\beta}_{H_1}$, respectively. Since the distribution of $T$ under $H_0$ is unclear, we use LOT to test \eqref{H0}.

First of all we need to estimate all the unknown parameters under $H_0$. Wu, Yang, and Liu (2014) showed that the nonnegative lasso estimator in \eqref{nnlasso} is
consistent under $H_0$. By Fan, Guo, and Hao (2012), a consistent estimator of $\sigma^2$ is $\hat{\sigma}^2=\|y-X\hat{\beta}_{LS}\|^2/n$, where $\hat{\beta}_{LS}$ is
the ordinary least squares estimator of $\beta$ under the submodel selected by the nonnegative lasso. Since $p$ is large, the neighborhood
${\mathcal{N}}(\hat{\beta}_{H_0},\hat{\sigma}^2)$ should be selected elaborately to avoid high-dimensional optimization. We select
${\mathcal{N}}(\hat{\beta}_{H_0},\hat{\sigma}^2)$ as
\begin{equation}{\mathcal{N}}(\hat{\beta}_{H_0,1})\times\cdots\times{\mathcal{N}}(\hat{\beta}_{H_0,p})
\times{\mathcal{N}}(\hat{\sigma}^2),\label{nbs}\end{equation}where $\hat{\beta}_{H_0,j}$'s are components of $\hat{\beta}_{H_0}$,
${\mathcal{N}}(\hat{\beta}_{H_0,j})=\{0\}$ for $\hat{\beta}_{H_0,j}=0$ and ${\mathcal{N}}(\hat{\beta}_{H_0,j})=\big[\hat{\beta}_{H_0,j}-\delta\hat{\sigma},\
\hat{\beta}_{H_0,j}+\delta\hat{\sigma}\big]$ otherwise, ${\mathcal{N}}(\hat{\sigma}^2)=[\hat{\sigma}^2-\delta,\ \hat{\sigma}^2+\delta]$, and $\delta>0$ is a constant. By
the importance sampling-based approach in Section \ref{subsec:is}, the $p$-value of the LOT for \eqref{H0} is given by \eqref{otpa}. Note that the asymptotic results in
Section \ref{sec:ap} cannot be dirrectly applied for diverging $p$. However, it is not hard to show that, if $H_0$ holds, then
$\P\big((\beta,\sigma^2)\in{\mathcal{N}}(\hat{\beta}_{H_0},\hat{\sigma}^2)\big)\to1$ as $n\to\infty$ under regularity conditions by selection consistency properties of
the nonnegative lasso (Wu, Yang, and Liu 2014). Therefore, similar to Theorem \ref{th: acr}, the asymptotic frequentist property of the LOT can be guaranteed.

We conduct a simulation study to compare the above LOT and the bootstrap test whose $p$-value is given in \eqref{bp}. In the simulation the rows of $X$ in \eqref{lm} are
i.i.d. from a multivariate normal distribution $N({0},{\Sigma})$ whose covariance matrix ${\Sigma}=(\sigma_{ij})_{p\times p}$ has entries $\sigma_{ii}=1,\ i=1,\ldots,p$
and $\sigma_{ij}=0.1,\ i\neq j$. The random errors $\varepsilon_1,\ldots,\varepsilon_n$ i.i.d. $\sim N(0,1)$. We use three configurations of $n$ and $p$,
$(n,p)=(20,40)$, $(n,p)=(40,80)$, and $(n,p)=(60,120)$. We take the tuning parameter $\lambda=4\sqrt{\log(p)/n}$ in the lasso and nonnegative lasso estimator recommended
by Wu, Yang, and Liu (2014). In the LOT, $\delta$ is set as 0.03 in \eqref{nbs}, and we compute the $p$-value in \eqref{nbpa} with 30 try points. Here the initial design
of the try points is an LHD, whose dimension is the number of non-zero $\hat{\beta}_{H_0,j}$'s; see \eqref{nbs}. In the two methods, the bootstrap sample sizes are both
2000. The significance levels $\alpha=0.05$ and $\alpha=0.1$ are considered.

{\begin{table}[htbp]\begin{center}\scriptsize\caption{\label{tab:nnlasso}Type I errors in Section \ref{subsec:nnlasso}} \centering\vspace{3mm} \vspace{3mm}
\begin{tabular}{llccccccccc} \multicolumn{11}{c}{$n=20,\ p=40$}\\[1mm]\hline &\quad\quad&\multicolumn{4}{c}{$\alpha=0.05$}&\quad&\multicolumn{4}{c}{$\alpha=0.1$}
\\[2mm]&&(i)&(ii)&(iii)&(iv)&&(i)&(ii)&(iii)&(iv)\\\hline
\\[-2mm]Bootstrap  &&0.084&0.094&0.096&0.092&&0.176&0.199&0.186&0.192
\\[1mm]LOT         &&0.048&0.056&0.056&0.050&&0.099&0.119&0.108&0.122
\\\hline \end{tabular}\vspace{3mm}
\begin{tabular}{llccccccccc} \multicolumn{11}{c}{$n=40,\ p=80$}\\[1mm]\hline &\quad\quad&\multicolumn{4}{c}{$\alpha=0.05$}&\quad&\multicolumn{4}{c}{$\alpha=0.1$}
\\[2mm]&&(i)&(ii)&(iii)&(iv)&&(i)&(ii)&(iii)&(iv)\\\hline
\\[-2mm]Bootstrap  &&0.134&0.172&0.160&0.168&&0.248&0.302&0.282&0.308
\\[1mm]LOT         &&0.052&0.062&0.062&0.060&&0.110&0.126&0.116&0.126
\\\hline \end{tabular}\vspace{3mm}
\begin{tabular}{llccccccccc} \multicolumn{11}{c}{$n=60,\ p=120$}\\[1mm]\hline &\quad\quad&\multicolumn{4}{c}{$\alpha=0.05$}&\quad&\multicolumn{4}{c}{$\alpha=0.1$}
\\[2mm]&&(i)&(ii)&(iii)&(iv)&&(i)&(ii)&(iii)&(iv)\\\hline
\\[-2mm]Bootstrap  &&0.206&0.216&0.194&0.239&&0.372&0.378&0.362&0.376
\\[1mm]LOT         &&0.066&0.054&0.060&0.051&&0.128&0.126&0.100&0.136
\\\hline \end{tabular}\end{center}\end{table}}

\begin{figure}[htbp]\begin{center}\scalebox{0.6}[0.6]{\includegraphics{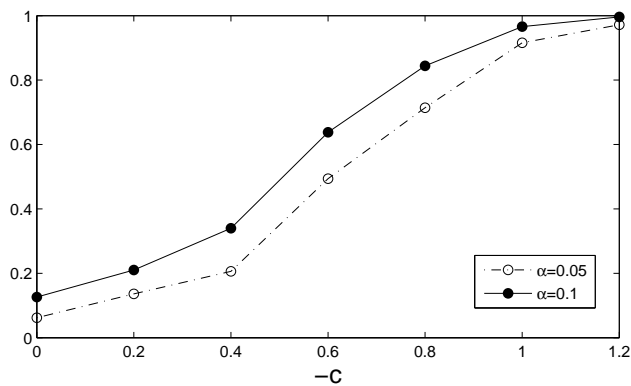}}\end{center}
\caption{Powers of the LOT in Section \ref{subsec:nnlasso} ($n=40,p=80$).}\label{fig:power}
\end{figure}

Four vectors of the coefficients under $H_0$ are used: (i)\ $\beta_1=\cdots=\beta_p=0$; (ii)\ $\beta_1=2$ and $\beta_j=0$ for other $j$; (iii)\ $\beta_1=\beta_2=2$ and
$\beta_j=0$ for other $j$; (iv)\ $\beta_1=\beta_2=\beta_3=2$ and $\beta_j=0$ for other $j$. To compute the power, we consider $\beta_1=2,\ \beta_2=c<0$, and $\beta_j=0$
for other $j$. For each model, we simulate 2000 data sets, and report the Type I errors and powers in Table \ref{tab:nnlasso} and Figure \ref{fig:power}, respectively.
It can be seen that the bootstrap test cannot control Type I error well, and that the LOT has reasonable performance in terms of Type I error and power. The power
performance of the LOT is similar for other parameter configurations.

\subsection{Interval estimation for the minimum of an unknown function}\label{subsec:nr} \hskip\parindent \vspace{-0.8cm}

Consider the nonparametric regression model \begin{equation}y=r(x)+\varepsilon,\label{nr}\end{equation}where $r$ is a continuous function defined on $[0,1]$ and
$\varepsilon\sim N(0,\sigma^2)$ is the random error. For given $x_1,\ldots,x_n\in[0,1]$, the responses are denoted by $y_1,\ldots,y_n$, respectively, where
$y_i=r(x_i)+\varepsilon_i$ and $\varepsilon_1,\ldots,\varepsilon_n$ are independent. Assume that $r$ has a unique minimum $\xi$ in $[0,1]$, i.e., $r(\xi)<r(x)$ for all
$x\in[0,1]$ with $x\neq \xi$. We are interested in constructing confidence intervals of $\xi=\xi(r)$. Without the random error, some related problems have been discussed
in the literature; see de Haan (1981) and de Carvalho (2011), among others. However, to the best of the author's knowledge, there is no result on interval estimation for
$\xi$ in the regression setting.

In model \eqref{nr}, the unknown parameter $r$ lies in an infinite-dimensional space. We shall show that, with a fixed design for $x_1,\ldots,x_n$, the problem of
construction confidence intervals for $\xi$ can be simplified to a finite-dimensional problem, and thus can be solved by the approaches in Section \ref{sec:imp}. Here we
only focus on the upper $1-\alpha$ confidence interval for $\alpha\in(0,1)$. Let $\hat{r}$ and $\hat{\sigma}$ be estimators of $r$ and $\sigma$. We use
$\hat{\xi}(\X)=\arg\min_{x\in[0,1]}\hat{r}(x)$ as an estimator of $\xi$ with $\X=(y_1,\ldots,y_n)'$, and consider the pivotal quantity $\xi(r)-\hat{\xi}(\X)$ with the
c.d.f. $H_{(r,\sigma)}(x)=\P\big(\xi(r)-\hat{\xi}(\X)\leqslant x\big)$. For $a\in[0,1],\ b_1,\ldots,b_n\in{\mathbb{R}}$, and $c>0$, let $\X^*=(y_1^*,\ldots,y_n^*)$
denote the set of independent random variables $y_i^*\sim N(b_i, c^2)$ for $i=1,\ldots,n$. Let $\theta=(a,b_1,\ldots,b_n,c)'$ and
$\tilde{H}_{\theta}(x)=\P\big(a-\hat{\xi}(\X^*)\leqslant x\big)$. Denote by $\mathcal{N}_n(\hat{\theta})\subset{\mathbb{R}}^{n+2}$ a neighborhood of
$\hat{\theta}=\big(\hat{\xi},\hat{r}(x_1),\ldots,\hat{r}(x_n),\hat{\sigma}\big)'$. Since $H_{(r,\sigma)}(x)=\tilde{H}_{(\xi(r),\, r(x_1),\ldots,r(x_n),\,\sigma)'}(x)$,
the following proposition is straightforward.

\begin{figure}[htbp]\begin{center}\scalebox{1}[1]{\includegraphics{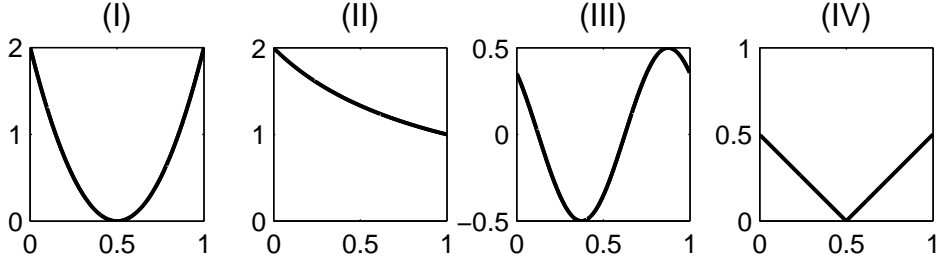}}\end{center} \caption{Four regression functions in simulations in Section
\ref{subsec:nr}.}\label{fig:funs}
\end{figure}

\begin{proposition}\label{prop: nr} For all $r$ and $\sigma$, if $\P\big((\xi,r(x_1),\ldots,r(x_n),\sigma)'\in\mathcal{N}_n(\hat{\theta})\big)\to1$, then
\begin{eqnarray*}\liminf_{n\to\infty}\P\left(\xi\leqslant\hat{\xi}+\sup_{\phi\in\mathcal{N}_n(\hat{\theta})}\tilde{H}_\phi^{-1}(1-\alpha)\right)
\geqslant 1-\alpha.\end{eqnarray*}\end{proposition}

\begin{table}[h]\begin{center}\scriptsize\caption{\label{tab:nr}Simulation results in Section \ref{subsec:nr}} \centering\vspace{3mm}
\begin{tabular}{llccccccc} \multicolumn{9}{c}{(I)}\\[1mm]\hline &\quad\quad&\multicolumn{3}{c}{$n=20$}&\quad&\multicolumn{3}{c}{$n=30$}
\\[2mm]&&CR&ML&SDL&&CR&ML&SDL\\\hline
\\[-2mm]Bootstrap  &&0.635&0.264&0.066&&0.643&0.243&0.060
\\[1mm]LOCI        &&0.933&0.487&0.062&&0.948&0.457&0.055
\\\hline \end{tabular}\vspace{3mm}
\begin{tabular}{llccccccc} \multicolumn{9}{c}{(II)}\\[1mm]\hline &\quad\quad&\multicolumn{3}{c}{$n=20$}&\quad&\multicolumn{3}{c}{$n=30$}
\\[2mm]&&CR&ML&SDL&&CR&ML&SDL\\\hline
\\[-2mm]Bootstrap  &&0.749&0.189&0.176&&0.766&0.161&0.159
\\[1mm]LOCI        &&0.935&0.281&0.220&&0.950&0.251&0.199
\\\hline \end{tabular}\vspace{3mm}
\begin{tabular}{llccccccc} \multicolumn{9}{c}{(III)}\\[1mm]\hline &\quad\quad&\multicolumn{3}{c}{$n=20$}&\quad&\multicolumn{3}{c}{$n=30$}
\\[2mm]&&CR&ML&SDL&&CR&ML&SDL\\\hline
\\[-2mm]Bootstrap  &&0.710&0.272&0.091&&0.656&0.233&0.069
\\[1mm]LOCI        &&0.973&0.573&0.127&&0.963&0.490&0.095
\\\hline \end{tabular}\vspace{3mm}
\begin{tabular}{llccccccc} \multicolumn{9}{c}{(IV)}\\[1mm]\hline &\quad\quad&\multicolumn{3}{c}{$n=20$}&\quad&\multicolumn{3}{c}{$n=30$}
\\[2mm]&&CR&ML&SDL&&CR&ML&SDL\\\hline
\\[-2mm]Bootstrap  &&0.721&0.486&0.203&&0.721&0.446&0.175
\\[1mm]LOCI        &&0.927&0.794&0.192&&0.957&0.787&0.161
\\\hline \end{tabular}\end{center}\end{table}

By Proposition \ref{prop: nr}, we can obtain LOCIs of $\xi$ which have the asymptotic frenquentist properties through optimizing the quantiles of $\tilde{H}_\phi$ over a
local region. In the following, let $\hat{r}$ be the Nadaraya-Watson estimator with kernel function $K$ and bandwidth $h$ (Hart 1997). Under regularity conditions,
$\sup_{x\in[0,1]}|\hat{r}(x)-r(x)|\to0$ in probability (H\"{a}rdle and Luckhaus 1984), which implies that $\hat{\xi}$ is a consistent estimator of $\xi$. Additionally, a
consistent estimator $\hat{\sigma}^2$ of $\sigma^2$ can be given from the residual sum of squares of $\hat{r}$. A choice of $\mathcal{N}_n(\hat{\theta})$ satisfying the
condition in Proposition \ref{prop: nr} is\begin{equation}\big[\hat{\xi}-\delta\hat{\sigma},\
\hat{\xi}+\delta\hat{\sigma}\big]\times\big[\hat{r}(x_1)-\delta\hat{\sigma},\
\hat{r}(x_1)+\delta\hat{\sigma}\big]\times\cdots\times\big[\hat{r}(x_n)-\delta\hat{\sigma},\ \hat{r}(x_n)+\delta\hat{\sigma}\big]\times\big[\hat{\sigma}-\delta,\
\hat{\sigma}+\delta\big],\label{Nnr}\end{equation}where $\delta>0$ is a constant.

We next conduct a simulation study to compare the bootstrap two-sided $1-\alpha$ confidence intervals and LOCIs with $\alpha=0.05$. Four regression functions in
\eqref{nr} are considered:
\begin{eqnarray*}&&\text{(I)}:\ r(x)=2(2x-1)^2;\ \ \text{(II)}:\ r(x)=2/(x+1);\\&&\text{(III)}:\ r(x)=\sin(2\pi x+3\pi/4)/2;
\ \ \text{(IV)}:\ r(x)=|x-1/2|;\end{eqnarray*}see Figure \ref{fig:funs}. For these functions, the values of $\xi$ are $1/2,\ 1,\ 3/8$, and $1/2$, respectively. We fix
$\sigma^2=1/4$, and $x_i=(2i-1)/(2n)$ for $i=1,\ldots,n$. The kernel function $K$ in $\hat{r}$ is the Epanechnikov kernel, and the bandwidth $h$ is set as $n^{-1/5}/5$.
In LOCIs, we use $\delta=0.25$ in \eqref{Nnr}, and take 60-run LHDs as the initial designs of try points for implementing neighborhood bootstrap. The bootstrap sample
size is 5000. Based on 5000 repetitions, we report the simulation results in Table \ref{tab:nr}. It can be seen that the bootstrap method performs poorly in terms of CR,
and that the LOCI is much better for all the cases.

\section{Discussion}\label{sec:diss}
\hskip\parindent \vspace{-0.8cm}

In this paper we have introduced local optimization-based inference including LOT and LOCI. The main advantage of our approach is that, unlike current frequentist
approach, it does not require hard work in deriving (asymptotic) distributions since its asymptotic frequentist properties hold as long as we have consistent estimators
of the unknown parameters. The implementation of our approach is based on standard computational methods such as importance sampling and Monte Carlo, which are easy to
master for practitioners. Local optimization-based inference can be viewed as an extended bootstrap that complements current frequentist inference. It can fast provide
frequentist solutions to complex problems in practice, and has broadly potential applications. Illustrative examples have shown these to some extent. Although local
optimization-based inference does not overshoot for regular problems (see Theorem \ref{th:aeb}), it is more suitable for non-regular problems in which the theoretical
derivation is difficult.

We give a further discussion on the specification of the neighborhood ${\mathcal{N}}(\hat{\theta})$ here. Generally speaking, the choice of ${\mathcal{N}}(\hat{\theta})$
is flexible; see Section \ref{sec:exm}. In real applications, for a dataset with fixed sample size $n$, it is not hard to find a proper ${\mathcal{N}}(\hat{\theta})$
that guarantees that LOT or LOCI has satisfactory performance via empirical evaluations. Besides the methods in Section \ref{subsec:N}, we can also use informative
priors, if any, to inform the construction of ${\mathcal{N}}(\hat{\theta})$. This provides a way to associate our approach with Bayesian statistics, and is valuable to
study in the future. A related problem to the specification of ${\mathcal{N}}(\hat{\theta})$ is that it is difficult to get the exact solution or to know how close an
approximate one to it even for a small ${\mathcal{N}}(\hat{\theta})$. This problem is not very serious in practice since our terminal is inference instead of
optimization. Simulation results in Section \ref{sec:exm} show that the design-based approximation with a moderate $L$ yields satisfactory finite-sample performance of
LOT and LOCI even for high-dimensional ${\mathcal{N}}(\hat{\theta})$. In fact, when bootstrap gives aggressive results, local optimization-based inference can always
improve its performance, even with a relatively poor optimization algorithm, since the corresponding optimization problem possesses a better solution principle (Xiong
2014): a better approximation to the exact solution yields less Type I error or higher coverage rate.

A disadvantage of local optimization-based inference is its computational cost. This can be viewed as the price of generality. We can replace the Monte Carlo method in
the implementation with LHD sampling or quasi Monte Carlo to improve the computational efficiency (Homem-de-Mello 2008). Iterative algorithms such as stochastic
approximation (Kushner and Yin 1997) are also available to solve the stochastic programming problem in \eqref{otp}. Another future topic is to apply the proposed
approach to general infinite-dimensional problems, which call for infinite-dimensional optimization techniques. In the neighborhood bootstrap method, we need to develop
new space-filling designs in infinite-dimensional spaces.

\section*{Appendix: Asymptotic properties of the design-based algorithm}\hskip\parindent \vspace{-0.8cm}

As mentioned in Section \ref{sec:imp}, $\max_{l=1,\ldots,L_n}H_{\phi_l}^{-1}(1-\alpha)$ can be used to approximate the upper limit
$\sup_{\phi\in\mathcal{N}_n(\hat{\theta})}H_\phi^{-1}(1-\alpha)$ in LOCI, where $\{\phi_1,\ldots,\phi_{L_n}\}$ is a dense subset of $\mathcal{N}_n(\hat{\theta})$. We
next prove frequentist properties of this approximation. These results are less important in practice since we can obtain an approximation as accurate as possible with a
powerful computer. We place them here because they may be still of interest in theory.

\begin{assumption}\label{as:dz} Let $\{a_n\}$ be a series of positive numbers. Denote
$\bar{H}_{\phi}(x)=\P\big(a_n[\xi(\phi)-\hat{\xi}(\X^*)]\leqslant x|\,\X\big)$, where $\X^*$ is drawn from $F(\cdot,\phi)$ given $\X$. As $n\to\infty$,
$$\max_{\phi\in\mathcal{N}_n(\hat{\theta})}\min_{l=1,\ldots,L_n} \big|\bar{H}_\phi^{-1}(1-\alpha)-\bar{H}_{\phi_l}^{-1}(1-\alpha)\big|=o_p(1).$$\end{assumption}

\begin{assumption}\label{as:ic}As $n\to\infty$ and $\delta\to0$, $\bar{H}_\theta\big(\bar{H}_\theta^{-1}(1-\alpha)-\delta\big)\to1-\alpha$.\end{assumption}

Note that the limits of $\bar{H}_{\varphi_1}$ and $\bar{H}_{\varphi_2}$ can be different for $\varphi_1,\ \varphi_2\in\mathcal{N}_n(\hat{\theta})$. Assumption
\ref{as:dz} requires that $\{\phi_1,\ldots,\phi_{L_n}\}$ should be dense enough so that any value of $\bar{H}_\phi^{-1}(1-\alpha)$ can be approximated accurately by some
element in $\{\bar{H}_{\phi_l}^{-1}(1-\alpha)\}_{l=1,\ldots,L_n}$. Assumption \ref{as:dz} holds under Assumptions \ref{as:nsz}-\ref{as:bc}, and relates to some
space-filling criterion (Johnson, Moore, and Ylvisaker 1990). Assumption \ref{as:ic} says that $\bar{H}_\theta$ is asymptotically continuous at
$\bar{H}_\theta^{-1}(1-\alpha)$. Under Assumption \ref{as:bc} (i), Assumption \ref{as:ic} holds.

\begin{theorem}\label{th:acr2}Under Assumptions \ref{as:n}, \ref{as:dz}, and \ref{as:ic},
\begin{eqnarray*}\liminf_{n\to\infty}\P\left(\xi\leqslant\hat{\xi}+\max_{l=1,\ldots,L_n}H_{\phi_l}^{-1}(1-\alpha)\right)
\geqslant 1-\alpha.\end{eqnarray*}\end{theorem}

\begin{proof} For any $n$, there exists $\theta_n^*\in\mathcal{N}_n(\hat{\theta})$ such that $\sup_{\phi\in\mathcal{N}_n(\hat{\theta})}\bar{H}_\phi^{-1}(1-\alpha)
<\bar{H}_{\theta_n^*}^{-1}(1-\alpha)+1/n$. Denote $l^*=\arg\min_{l=1,\ldots,L_n} \big|\bar{H}_{\phi_l}^{-1}(1-\alpha)-\bar{H}_{\theta_n^*}^{-1}(1-\alpha)\big|$. We have
$\max_{l=1,\ldots,L_n}\bar{H}_{\phi_l}^{-1}(1-\alpha)\geqslant\bar{H}_{\phi_{l^*}}^{-1}(1-\alpha)\geqslant\bar{H}_{\theta_n^*}^{-1}(1-\alpha)
-|\bar{H}_{\phi_{l^*}}^{-1}(1-\alpha)-\bar{H}_{\theta_n^*}^{-1}(1-\alpha)|$. Therefore,
\begin{eqnarray*}&&\P\left(\xi\leqslant\hat{\xi}+\max_{l=1,\ldots,L_n}H_{\phi_l}^{-1}(1-\alpha)\right)=\P\left(a_n(\xi-\hat{\xi})\leqslant
\max_{l=1,\ldots,L_n}\bar{H}_{\phi_l}^{-1}(1-\alpha)\right)
\\&&\geqslant \P\left(a_n(\xi-\hat{\xi})\leqslant \bar{H}_{\theta_n^*}^{-1}(1-\alpha)
-|\bar{H}_{\phi_{l^*}}^{-1}(1-\alpha)-\bar{H}_{\theta_n^*}^{-1}(1-\alpha)|\right)
\\&&\geqslant \P\left(a_n(\xi-\hat{\xi})\leqslant\sup_{\phi\in\mathcal{N}_n(\hat{\theta})}\bar{H}_\phi^{-1}(1-\alpha)+o_p(1)-1/n\right)
\\&&\geqslant \P\left(a_n(\xi-\hat{\xi})\leqslant\bar{H}_\theta^{-1}(1-\alpha)+o_p(1)\right)-\P\left(\theta\in\mathcal{N}_n(\hat{\theta})\right)
\\&&=1-\alpha+o(1)-\P\left(\theta\in\mathcal{N}_n(\hat{\theta})\right)\to1-\alpha,
\end{eqnarray*}which completes the proof.\end{proof}

The following theorem is straightforward.

\begin{theorem}\label{th:aeb2}Under Assumptions \ref{as:n}--\ref{as:bc}, for all $\theta\in\Theta$ and $\alpha\in(0,1)$,
\begin{eqnarray*}\lim_{n\to\infty}\P\left(\xi\leqslant\hat{\xi}+\max_{l=1,\ldots,L_n}H_{\phi_l}^{-1}(1-\alpha)\right)
=1-\alpha.\end{eqnarray*}\end{theorem}

\section*{Acknowledgements}
\hskip\parindent \vspace{-0.8cm}

This work is supported by the National Natural Science Foundation of China (Grant No. 11271355, 11471172). The author is grateful to the support of Key Laboratory of
Systems and Control, Chinese Academy of Sciences.

 \vspace{1cm} \noindent{\Large\bf References}

{\begin{description}

\item{}
Bickel, P. J. and G\"{o}tze F, van Zwet, W. R. (1997), Resampling fewer than n observations gains, losses, and remedies for losses. \textit{Statistica Sinica}, 7: 1--31.

\item{}
Bickel, P. J. and Ren, J-J. (2001), The bootstrap in hypothesis testing, Lecture Notes-Monograph Series, Vol. 36, \textit{State of the Art in Probability and
Statistics}, 91--112.

\item{}
Blischke, W. R. (1974), On nonregular estimation. II. Estimation of the location parameter of the gamma and Weibull distributions. \textit{Communications in Statistics},
3, 1109--1129.

\item{}
Boyd, S. and Vandenberghe, L. (2004), \textit{Convex Optimization}. Cambridge University Press, Cambridge.

\item{}
Cheng, R. C. H. and Amin, N. A. K. (1983), Estimating parameters in continuous univariate distributions with a shifted origin. \textit{Journal of the Royal Statistical
Society, Ser. B}, {45}: 394--403.


\item{}
Cousineau, D. (2009), Fitting the three-parameter Weibull distribution: Review and evaluation of existing and new methods. \textit{IEEE Transactions on Dielectrics and
Electrical Insulation}, 16: 281--288.

\item{}
Cox, D. R. and Oakes, D. (1984), \textit{Analysis of Survival Data}. CRC Press.

\item{}
Davison, A. C. and Hinkley, D. V. (1997), \textit{Bootstrap Methods and Their Application}. Cambridge University Press, Cambridge.

\item{}
de Carvalho, M. (2011), Confidence intervals for the minimum of a function using extreme value statistics. \textit{International Journal of Mathematical Modelling \&
Numerical Optimisation}, 2: 288--296.

\item{}
de Haan, L. (1981), Estimation of the minimum of a function using order statistics. \textit{Journal of the American Statistical Association}, 76: 467--469.

\item{}
Efron, B. (1979), Bootstrap methods: Another look at the jackknife. \textit{Annals of Statistics}, 7: 1--26.

\item{}
Efron, B., Hastie, T., Johnstone, L., and Tibshirani, R. (2004), Least angle regression (with discussion). \textit{Annals of Statistics}, 32: 407--451.

\item{}
Ethier, S. N. (1982), Testing for favorable numbers on a roulette wheel. \textit{Journal of the American Statistical Association}, 77: 660--665.

\item{}
Fan, J. Guo, S., and Hao, N. (2012), Variance estimation using refitted cross-validation in ultrahigh dimensional regression. \textit{Journal of the Royal Statistical
Society, Ser. B}, {74}: 37--65.

\item{}
Fan, J. and Li, R. (2001), Variable selection via nonconcave penalized likelihood and its oracle properties. \textit{Journal of the American Statistical Association},
{96}: 1348--1360.

\item{}
Fan, J., Zhang, C., and Zhang J. (2001), Generalized likelihood ratio statistics and Wilks phenomenon. \textit{Annals of statistics}, 29: 153--193.

\item{}
Fang, K. T., Hickernell, F. J., and Winker, P. (1996), Some global optimization algorithms in statistics, in \textit{Lecture Notes in Operations Research}, eds. by Du,
D. Z., Zhang, X. S. and Cheng, K. World Publishing Corporation, 14--24.

\item{}
Fang, K. T., Lin, D. K. J., Winker, P., and Zhang, Y. (2000), Uniform design: theory and application. \textit{Technometrics}, 42: 237--248.

\item{}
Fisher, R. A. (1959), \textit{Statistical Methods and Scientific Inference}, 2nd ed., revised. Hafner Publishing Company, New York.

\item{}
Frank, L. E. and Heiser, W. J. (2008), Feature selection in feature network models: Finding predictive subsets of features with the positive Lasso. \textit{British
Journal of Mathematical and Statistical Psychology}, 61: 1--27.

\item{}
Gelfand, A. E., Glaz, J., Kuo, L., and Lee, T. M. (1992), Inference for the maximum cell probability under multinomial sampling. \textit{Naval Research Logistics}, 39:
97--114.

\item{}
Ghosh, J. K., Delampady, M., and Samanta, T. (2007), \textit{An Introduction to Bayesian Analysis: Theory and Methods}. Springer.

\item{}
Glaz, J. and Sison, C. P. (1999), Simultaneous confidence intervals for multinomial proportions. \textit{Journal of Statistical Planning and Inference}, 82: 251--262.

\item{}
Hall, P. (1992), \textit{The bootstrap and Edgeworth Expansion}, Springer, New York.

\item{}
H\"{a}rdle, W. and Luckhaus, S. (1984), Uniform consistency of a class of regression function estimators. \textit{Annals of Statistics}, 12: 612--623.

\item{}
Hart, J. D. (1997), \textit{Nonparametric Smoothing and Lack-of-Fit Tests}, Springer, New York.

\item{}
Homem-de-Mello T. (2008), On rates of convergence for stochastic optimization problems under non-independent and identically distributed sampling. \textit{SIAM Journal
on Optimization}, 19: 524--551.

\item{}
Johnson, M. E., Moore, L. M., and Ylvisaker, D. (1990), Minimax and maximin distance designs. \textit{Journal of statistical planning and inference}, 26: 131--148.

\item{}
Kushner, H. J. and Yin, G. G. (1997), \textit{Stochastic Approximation Algorithms and Applications}. Springer, New York.

\item{}
Lehmann, E. L. and Romano, J. P. (2006), \textit{Testing statistical hypotheses}, 3rd., Springer, Science \& Business Media.

\item{}
Lockhart, R. A. and Stephens, M. A. (1994), Estimation and tests of fit for the three-parameter Weibull distribution. \textit{Journal of the Royal Statistical Society,
Ser. B}, 56: 491--500.

\item{}
McKay, M. D., Beckman, R. J., and Conover, W. J. (1979), A comparison of three methods for selecting values of input variables in the analysis of output from a computer
code, \textit{Technometrics}, 21: 239--245.

\item{}
Meng, X-L. (1994), Posterior predictive $p$-values. \textit{Annals of Statistics}, 22: 1142--1160.

\item{}
Murthy, D. N. P., Xie, M., and Jiang, R. (2004), \textit{Weibull Models}. John Wiley \& Sons.

\item{}
Park, J. S. (2001), Optimal Latin-hypercube designs for computer experiments. \textit{Journal of Statistical Planning Inference}, 39: 15--111.

\item{}
Patil, G. P. and Taillie, C. (1979), An overview of diversity. \textit{Ecological Diversity in Theory and Practice}, ed. by Grassle, J. F. et al. International
Co-operative Publishing House. Fairland, MD.

\item{}
Politis, D. N., Romano, J. P., and Wolf, M. (1999), \textit{Subsampling}. Springer, New York.

\item{}
Shao, J. and Tu, D. (1995), \textit{The Jackknife and Bootstrap}. Springer, New York.

\item{}
Shapiro, A. (2003), Monte Carlo sampling methods, in Stochastic Programming. \textit{Handbook in OR \& MS}, Vol. 10, ed. by Ruszczy\'{n}ski, A.  and Shapiro, A.,
North-Holland, Amsterdam.

\item{}
Silvapulle, M. J. and Sen, P. K. (2011), \textit{Constrained Statistical Inference: Order, Inequality, and Shape Constraints}. John Wiley \& Sons.


\item{}
Teimouri, M., Hoseini, S. M., and Nadarajah, S. (2013), Comparison of estimation methods for the Weibull distribution. \textit{Statistics: A Journal of Theoretical and
Applied Statistics}, 47: 93--109.

\item{}
Tibshirani, R. (1996), Regression shrinkage and selection via the lasso. \textit{Journal of the Royal Statistical Society, Ser. B}. 58: 267--288.


\item{}
Wu, L., Yang, Y., and Liu, H. (2014), Nonnegative-lasso and application in index tracking. \textit{Computational Statistics} \& \textit{Data Analysis}, 70: 116--126

\item{}
Xiong, S. (2014), Better solution principle: A facet of concordance between optimization and statistics. arXiv preprint.

\item{}
Xiong, S. and Li, G. (2009), Inference for ordered parameters in multinomial distributions. \textit{Science in China, Ser. A: Mathematics}, 52: 526--538.

\item{}
Xiong, S. and Li, G. (2010), Dispersion comparisons of two probability vectors under multinomial sampling. \textit{Acta Mathematica Scientia}, 30: 907--918.

\item{}
Zhang. C-H. (2010), Nearly unbiased variable selection under minimax concave penalty. \textit{Annals of Statistics}, 38: 894--942.

\end{description}}

\end{document}